\documentclass[12pt,leqno]{article}

\input{style}  

\usepackage[colorlinks]{hyperref}
\usepackage[colorinlistoftodos, textwidth=4cm, shadow]{todonotes}
\usepackage{amsmath}
\usepackage[displaymath, tightpage]{preview}

\begin{document}
%%%%%%%%%%%%%%%%%%%%%%%%%%%%%%%%%%%%%%%%%%%%%%%%%%%%%%%%%%%%%%%%%%

\title{Optimal Order Scheduling for Deterministic Liquidity Patterns}
 
\author{Peter Bank and Antje Fruth\\ Technische
  Universit{\"a}t Berlin\\
  Institut f{\"u}r Mathematik\\ Stra{\ss}e des 17. Juni 136, 10623
  Berlin, Germany \\ (bank@math.tu-berlin.de) % \and 
  % Antje Fruth
  % \\ \todo[inline]{Add contact Details}
}
\date{\today}

\maketitle

\begin{abstract}
  We consider a broker who has to place a large order which consumes a
  sizable part of average daily trading volume. The broker's aim is
  thus to minimize execution costs he incurs from the adverse impact
  of his trades on market prices. By contrast to the previous
  literature, see, e.g., \citet{ObizhaevaWang},
  \citet{PredoiuShaikhetShreve}, we allow the liquidity parameters of
  market depth and resilience to vary deterministically over the
  course of the trading period. The resulting singular optimal control
  problem is shown to be tractable by methods from convex analysis
  and, under minimal assumptions, we construct an explicit solution to
  the scheduling problem in terms of some concave envelope of the
  resilience adjusted market depth.
\end{abstract}
 
\begin{description}
\item[Keywords:] Order scheduling, liquidity, convexification,
  singular control, convex analysis, envelopes, optimal order execution
%\item[JEL Classification:] \todo{Add}
%\item[AMS Subject Classification (1991):] \todo{Add}
\end{description}

%\newpage

%\listoftodos

%\newpage

\section{Introduction}
\label{sec:introduction}

It is well-known that market liquidity exhibits deterministic intraday
patterns; see, e.g., \citet{ChordiaEtal} or \citet{KempfMayston} for
some empirical investigations. The academic literature on optimal
order scheduling, however, mostly considers time-invariant
specifications of market depth and resilience; cf.
\citet{ObizhaevaWang}, \citet{AlfonsiFruthSchied},
\citet{PredoiuShaikhetShreve}. It thus becomes an issue how to account
for time-varying specifications of these liquidity parameters when
minimizing the execution costs of a trading schedule.

Using dynamic programming techniques and calculus of variations, this
problem was addressed by \citet{FruthUrusovSchoeneborn}. These authors
show that under certain additional assumptions on these patterns there
is a time-dependent level for the ratio of the number of orders still
to be scheduled and the current market impact which signals when
additional orders should be placed. Explicit solutions are provided
for some special cases where the broker is continually issuing
orders. The thesis \cite{Fruth} discusses conditions under which the
order signal structure persists in case of stochastically varying
liquidity parameters. \citet{AcevedoAlfonsi} use backward induction
arguments in discrete time and then pass to continuous time to compute
optimal policies for nonlinear specifications of market impacts which
are scaled by a time-dependent factor satisfying some strong
regularity conditions. In their approach order schedules are allowed
in principle to sell and buy along the way, regardless of the sign of
the desired terminal position, and they proceed to identify conditions
(deemed to ensure absence of market manipulation strategies) under
which optimal schedules will not do so. Optimal schedules are then
obtained only under a strong assumption linking resilience and market
depth to each other along with their time derivatives.

By contrast to these approaches, we focus from the outset on pure
buying or selling schedules and show how to reduce our optimization
problem to a convex one. Hence, we do not have to impose conditions
ensuring that orders are scheduled in certain ways at certain
times. Instead, optimal order sizes and times are derived endogenously
from the structure of market depth and resilience alone. This is made
possible by the use of convex analytic first-order characterizations
of optimality which we show are intimately related to the construction
of generalized concave envelopes of a resilience-adjusted form of
market depth.  Under minimal assumptions, this allows us to
characterize when optimal schedules exist and, if so, to construct
them explicitly in terms of these envelopes. We illustrate our
findings by recovering the analytic solution of \citet{ObizhaevaWang}
and we show how optimal schedules depend on fluctuations in market
depth and the level of resilience. It turns out that with time-varying
market depth optimal order schedules do not have to consist of big
initial and terminal trades with infinitesimal ones in between as
typically found in the previous literature. We also find that lower
resilience will let optimal schedules focus more on (local) maxima of
market depth to the extent that with no resilience optimal schedules
trade only when market depth is at its global maximum.

\section{Setup}
\label{sec:setup}

We consider a broker who has to place an order of a total number of
$x>0$ shares of some stock. The broker knows that, due to limited
liquidity of the stock, these orders will be executed at a mark-up
over some reference stock price. This mark-up will depend on the
broker's past and present trades. For our specification of the mark-up
we adopt the model proposed by \citet{ObizhaevaWang}, see also
\citet{AlfonsiFruthSchied} and \citet{PredoiuShaikhetShreve} for
further motivation of this approach. By contrast to these papers, but
in line with \citet{FruthUrusovSchoeneborn} and
\citet{AcevedoAlfonsi}, we will allow for the market's liquidity
characteristics of depth and resilience to be changing over time
according to a deterministic pattern.

Specifically, given the broker's cumulative purchases $X=(X_t)_{t \geq
  0}$, a right-continuous increasing process with $X_{0-} \set 0$, the
resulting mark-up evolves according to the dynamics
\begin{equation}
  \label{eq:1}
 \eta^X_{0-} \set \eta_0 \geq 0, \quad d\eta^X_t = \frac{dX_t}{\delta_t}-r_t \eta^X_t \,dt 
\end{equation}
where $\delta_t$ describes the market's depth at time $t \geq 0$ and
where $r_t$ measures its current resilience. Thus, in our model market
impact is taken to be a linear function of order size, the slope at
any one time being determined by the market depth. Moreover, market
impact decays over time at the rate specified by the market's
resilience.

Clearly, \eqref{eq:1} has the right-continuous solution
\begin{equation}
  \label{eq:2}
  \eta^X_t \set \left(\eta_0+\int_{[0,t]}  \frac{\rho_s}{\delta_s} \,dX_s\right)/\rho_t \mtext{with} \rho_t \set
    \exp\left(\int_0^t r_s \,ds\right), \quad t \geq 0, 
\end{equation}
under
\begin{Assumption}
  \label{as:1}
  The resilience pattern is given by a strictly positive and locally
  Lebesgue-integrable function $r:[0,\infty) \to (0,\infty)$.
\end{Assumption}

In the sequel we shall require furthermore

\begin{Assumption}
	\label{as:2}
	The pattern of market depth $\delta:[0,\infty) \to [0,\infty)$
        is nonnegative, not identically zero, bounded and
        upper-semicontinuous with $\limsup_{t \uparrow \infty}
        \delta_t/\rho_t=0$.
\end{Assumption}

The broker's aim is to minimize the cumulative mark-up costs:
\begin{equation}
  \label{eq:3}
 \text{Minimize} \;  C(X) \set \int_{[0,\infty)} \left(\eta^X_{t-} + \frac{\Delta_t
      X}{2\delta_t}\right) \,dX_t \mtext{subject to} X \in \mathcal{X}
\end{equation}
where $\Delta_t X \set X_{t+}-X_{t-}$ and
\begin{displaymath}
  \mathcal{X} \set \descr{(X_t)_{t \geq 0} \mtext{right-cont., incr.}}{X_{0-}=0, X_\infty=x, C(X)<\infty}
\end{displaymath}
with the notation $X_\infty \set \lim_{t \uparrow \infty} X$.

\begin{Remark}
\begin{enumerate}
\item Note that the $\frac{\Delta_t X}{2\delta_t}$-term
  in~\eqref{eq:3} accounts for the costs a non-infinitesimal order
  will incur due to its own mark-up effect; cf., e.g.,
  \citet{AlfonsiFruthSchied} or \citet{PredoiuShaikhetShreve} who in
  addition show how costs functionals as in~\eqref{eq:3} emerge with
  stochastic reference prices evolving as martingales when the broker
  is risk-neutral. Note also that, since we let $X_{0-} \set 0$, a
  value of $X_0>0$ corresponds to an initial jump of size $\Delta_0 X
  = X_0$ in the order schedule.

\item To impose liquidation over a finite time horizon $T \geq 0$, one
  merely has to let the market depth $\delta_t =0$ for $t>T$. Indeed,
  following the convention that $1/0=\infty$ in the
  integration~\eqref{eq:2}, $\eta^X$ and thus the costs $C(X)$ will
  then be infinite for any order schedule $X$ which increases after
  $T$.

\item Strict positivity of $r$ comes without loss of generality since
  if resilience $r = 0$ vanishes almost everywhere on an interval
  $[t_0,t_1]$ there is no need to trade it off against market depth
  there and it is optimal to trade whatever amount is to be traded at
  the moment(s) when market depth $\delta$ attains its maximum over
  this period; cf. Proposition~\ref{pro:3}. Hence, $\delta$ could be
  assumed to take this maximum value at $t_0$ and the interval
  $(t_0,t_1]$ then be removed from consideration.

\item The assumption of upper-semicontinuous market depth $\delta$ is
  necessary to rule out obvious counterexamples for existence of
  optimal schedules. For unbounded upper-semicontinuous $\delta$ one
  can easily show that $\inf_\mathcal{X} C=x \eta_0/\rho_\infty$, and
  so there is no optimal schedule.  The $\limsup$-condition is needed
  to rule out the optimality of deferring part of the order
  indefinitely.

\item Including a discount factor with locally Lebesgue-integrable
  discount rate $\bar{r} = (\bar{r}_t) \geq 0$ in our mark-up costs is
  equivalent to considering $\widetilde{\delta}_t \set \delta_t
  \exp(\int_0^t \bar{r_s} \,ds)$ and $\widetilde{r}_t \set
  r_t+\bar{r}_t$, $t \geq 0$ instead of $\delta$ and $r$ above.
\end{enumerate}  
\end{Remark}

\section{Main result and sketch of its proof}
\label{sec:main-results}

The main result of this paper is the solution to
problem~\eqref{eq:3}. It describes up to what mark-up level our broker
should be willing to place orders at any point in time in order to
minimize mark-up costs:
\begin{Theorem}
  \label{thm:0}
  Suppose Assumptions~\ref{as:1} and~\ref{as:2} hold, let $\lambda_t
  \set \delta_t/\rho_t$, $\widetilde{\lambda}_t \set \sup_{u \geq t}
  \lambda_u$ and define
 \begin{equation}
   \label{eq:4}
   L^*_t = 
   \inf_{u > t}
   \frac{\widetilde{\lambda}_u-\widetilde{\lambda}_t}{\widetilde{\lambda}_u/\rho_u-\widetilde{\lambda}_t/\rho_t}, \quad t \geq 0\,,
 \end{equation}
 where we follow the convention that $0/0 \set 0$.

 Then the optimal order schedule strategy is to place orders at any
 time $t \geq 0$ if and while the resulting mark-up is no larger than
 $y^*L^*_t/\rho_t$, i.e.,
 \begin{equation}
   \label{eq:5}
  		X^*_t = \lambda_0(y^* L^*_0-\eta_0)^++ \int_{(0,t]} \lambda_s \,d
                \sup_{0 \leq v \leq s}\left\{(y^* L^*_v) \vee
                  \eta_0\right\}\,, \quad t \geq 0,
 \end{equation}
 provided the constant $y^*>0$ in~\eqref{eq:5} can be chosen such that
 $X^*_\infty = x$. This is the case if and only if the right side
 of~\eqref{eq:5} with $y^* \set 1$ remains bounded as $t \uparrow
 \infty$. If this is not the case, we have $\inf_{X \in \mathcal{X}}
 C(X)=0$ and the problem does not have a solution.
\end{Theorem}

The following results outline the proof of this theorem and may be of
independent interest. Our first auxiliary result provides a
mathematically more convenient formulation of problem~\eqref{eq:3}:

\begin{Proposition}
  \label{pro:0}
  Suppose Assumptions~\ref{as:1} and~\ref{as:2} hold, let $\lambda
  \set \delta/\rho$, $\kappa \set \lambda/\rho=\delta/ \rho^2$ and
  define, for increasing and right-continuous $Y=(Y_t)_{t \geq 0}$,
 \begin{displaymath}
   K(Y) \set \frac{1}{2}\int_{[0,\infty)} \kappa_t \,d(Y^2_t)\,.
 \end{displaymath}
 Then
  \begin{equation}
    \label{eq:6}
    Y_t = \eta_0+\int_{[0,t]} \frac{dX_s}{\lambda_s}, \, Y_{0-} \set \eta_0, \mtext{and}
    X_t = \int_{[0,t]} \lambda_s \,dY_s, \, X_{0-} \set 0, \quad t \geq 0,
  \end{equation}
  define mappings from $\mathcal{X}$ to
  \begin{displaymath}
  \mathcal{Y} \set \descr{(Y_t)_{t \geq 0}
  \mtext{right-cont., incr.}}{Y_{0-} \set \eta_0, \int_{[0,\infty)} \lambda_t \,dY_t=x, K(Y)<\infty}
  \end{displaymath}
  and vice versa such that
 \begin{displaymath}
   C(X) = K(Y) \,.
 \end{displaymath}
\end{Proposition}

As a result, with these choices of $\kappa$ and $\lambda$,
optimization problem~\eqref{eq:3} is equivalent to the following
problem:
 \begin{equation}
  \label{eq:7}
 \text{Minimize} \;  K(Y) \set \frac{1}{2}\int_{[0,\infty)} \kappa_t
 \,d(Y^2_t) \mtext{subject to} Y \in \mathcal{Y} \,.
\end{equation}
Neither problem~\eqref{eq:3} nor problem~\eqref{eq:7} is convex in
general:

\begin{Proposition}
  \label{pro:1}
  For upper-semicontinuous $\kappa$, the functional $K=K(Y)$
  of~\eqref{eq:7} is (strictly) convex for right-continuous,
  increasing $Y$ with $Y_{0-}=\eta_0$ if and only if $\kappa$ is
  (strictly) positive and (strictly) decreasing.
\end{Proposition}

Convexity can always be arranged for, though, in the following sense:

\begin{Theorem}
  \label{thm:1}
  Let $\lambda$, $\kappa$ be as in Proposition~\ref{pro:0}. Then
  optimization problem~\eqref{eq:7} has the same value as the convex
  optimization problem
 \begin{equation}
  \label{eq:8}
  \text{Minimize} \; \widetilde{K}(\widetilde{Y}) \set 
  \frac{1}{2}\int_{[0,\infty)} \widetilde{\kappa}_t
  \,d(\widetilde{Y}^2_t)
  \mtext{subject to} \widetilde{Y} \in \widetilde{\mathcal{Y}}
\end{equation}
where $\widetilde{\kappa}_t \set \widetilde{\lambda}_t/\rho_t$,
$\widetilde{\lambda}_t \set \sup_{u \geq t} \lambda_u$, $t \geq 0$,
and
\begin{displaymath}
  \widetilde{\mathcal{Y}} \set \descr{(\widetilde{Y}_t)_{t \geq 0}
  \mtext{right-cont., incr.}}{\widetilde{Y}_{0-} \set \eta_0, \int_{[0,\infty)} \widetilde{\lambda}_t \,d\widetilde{Y}_t=x, \widetilde{K}(\widetilde{Y})<\infty}\,.
\end{displaymath}
Moreover, any solution $\widetilde{Y}^*$ to~\eqref{eq:8} with
$\{d\widetilde{Y}^*>0\}\subset\{\widetilde{\lambda}=\lambda\}$ will
also be a solution to~\eqref{eq:7}.
\end{Theorem}
\begin{Remark}
  For an increasing process $Y=(Y_t)_{t \geq 0}$ we say that $t$ is a
  point of increase towards the right and write $dY_t>0$ if
  $Y_{t-}<Y_u$ for any $u>t$. A similar convention applies to
  decreasing processes and points of decrease towards the right.
\end{Remark}

The next proposition describes the (necessary and sufficient)
first-order conditions for optimality in problem~\eqref{eq:8}. As one
would expect, the broker has to strike a balance between the impact of
current orders on future mark-up costs (as represented by the left
side of~\eqref{eq:9} below) and the current prospect on future market
conditions (as represented by the decreasing envelope
$\widetilde{\lambda}$ of market depth over resilience on the right
side of that equation):

\begin{Proposition}
  \label{pro:2}
  For $\widetilde{\kappa}$, $\widetilde{\lambda} \geq 0$ as in
  Theorem~\ref{thm:1}, $\widetilde{Y}^* \in \widetilde{\mathcal{Y}}$
  solves~\eqref{eq:8} if and only if there is a constant $y>0$ such
  that
  \begin{equation}
    \label{eq:9}
    -\int_{[t,\infty)} \widetilde{Y}^*_u
    \,d\widetilde{\kappa}_u \geq y
    \widetilde{\lambda}_t \mtext{for $t \geq 0$ with `$=$' whenever} d\widetilde{Y}^*_t>0\,.
  \end{equation}
\end{Proposition}

Constructing right-continuous increasing $\widetilde{Y}^* \geq 0$
satisfying the first order conditions of~\eqref{eq:9} can be done by
using a time-change and concave envelopes; see also Figure~\ref{fig:0}
below:

\begin{Theorem}
  \label{thm:2} Under Assumptions~\ref{as:1} and~\ref{as:2}, consider
  the level passage times ${\tau}_k \set \inf \descr{t \geq
    0}{\widetilde{\kappa}_t \leq k}$ and let $\widetilde{\Lambda}_k \set
  k \rho_{\tau_k}$, $k \in (0,\widetilde{\kappa}_0]$ and
  $\widetilde{\Lambda}_0 \set 0$. 

  Then $\widetilde{\Lambda}$ is a continuous increasing map on
  $[0,\widetilde{\kappa}_0]$. Its concave envelope
  $\widehat{{\Lambda}}$ is absolutely continuous with a
  left-continuous, decreasing density $\partial
  \widehat{{\Lambda}}=(\partial \widehat{\Lambda}_k)_{0 < k \leq
    \widetilde{\kappa}_0} \geq 0$. Moreover, letting $\partial
  \widehat{{\Lambda}}_0 \set \partial \widehat{{\Lambda}}_{0+}$, we
  have that for any $y>0$ and $\eta_0\geq0$, $\widetilde{Y}^*_t \set
  (y\partial \widehat{{\Lambda}}_{\widetilde{\kappa}_t}) \vee \eta_0$,
  $t \geq 0$, with $\widetilde{Y}^*_{0-}\set \eta_0$ yields a
  right-continuous increasing process satisfying~\eqref{eq:9}.
\end{Theorem}

Combining the previous results, we shall obtain the following solution
to our original problem~\eqref{eq:3} which also provides a
characterization different from that outlined in Theorem~\ref{thm:0};
see also Figure~\ref{fig:01} below:

\begin{Corollary}
  \label{cor:1}
  Under the assumptions of Theorem~\ref{thm:2} and using its notation
  we have the following dichotomy:

  In case $|\partial \widehat{\Lambda}|_{\mathbf{L}^2} \set
  (\int_{0}^{\widetilde{\kappa}_0} (\partial \widehat{\Lambda}_k)^2
  \,dk)^{\frac{1}{2}}<\infty$ we can choose $y^*>0$ uniquely such that
 \begin{equation}
   \label{eq:10}
   X^*_t \set \lambda_0 (y^*\partial
     \widehat{\Lambda}_{\widetilde{\kappa}_0}-\eta_0)^+ 
      + \int_{(0,t]} \lambda_s \,d\left\{ (y^* \partial \widehat{\Lambda}_{\widetilde{\kappa}_s}) \vee
     \eta_0\right\}, \quad t \geq 0,
 \end{equation}
 increases from $ X^*_{0-} \set 0$ to $X^*_\infty = x$; this $X^* \in
 \mathcal{X}$ is an optimal order schedule for
 problem~\eqref{eq:3}. In the special case where $\eta_0=0$, $y^* =
 x/|\partial \widehat{\Lambda}|_{\mathbf{L}^2}^2$ and the minimal
 costs are given by $C(X^*) = x^2/(2|\partial
 \widehat{\Lambda}|_{\mathbf{L}^2}^2)$.

 If, by contrast, $|\partial \widehat{\Lambda}|_{\mathbf{L}^2}=\infty$
 then we have $\inf_{X \in \mathcal{X}} C(X) = 0$ and
 problem~\eqref{eq:3} does not have a solution.
\end{Corollary}

\section{Illustrations}
\label{sec:illustrations}

Corollary~\ref{cor:1} reduces the construction of optimal order
schedules to the computation of a concave envelope. This can often be
done in closed form, see, e.g., our treatment in
Section~\ref{sec:ObizhaevaWang} of the constant parameter case from
\citet{ObizhaevaWang}. Alternatively, one can resort to highly
efficient numerical methods from discrete geometry to come up with
solutions to essentially arbitrary liquidity patterns as we illustrate in
Section~\ref{sec:Numerics}.

\subsection{Constant market depth and resilience}\label{sec:ObizhaevaWang}

Let us first show how to recover the solution of \citet{ObizhaevaWang}
who consider a time horizon $T>0$ and constant market depth $\delta_t
\equiv \delta_0 1_{[0,T]}(t)$ and constant market resilience $r_t
\equiv r_0>0$, $t\geq 0$. In this case we have
\begin{align*}
\lambda_t &= \widetilde{\lambda}_t =
\delta_0 e^{-r_0 t} 1_{[0,T]}(t) \mtext{and}
\kappa_t = \widetilde{\kappa}_t = \delta_0 e^{-2r_0 t}1_{[0,T]}(t)\,.
\end{align*}
Hence,
\begin{align*}
\rho_{\tau_k} = \sqrt{\delta_0/(k \vee \kappa_T)} \mtext{and} 
 \widetilde{\Lambda}_k = \sqrt{\delta_0 k} \wedge
 (\sqrt{\delta_0/\kappa_T}k), 
\quad 0 \leq k \leq \delta_0.
\end{align*}
Thus, $\widetilde{\Lambda}$ is its own concave envelope, i.e.,
$\widetilde{ \Lambda} = \widehat{\Lambda}$, and its left-continuous
density is
$$
\partial \widehat{\Lambda}_k = \begin{cases} \frac{1}{2}
  \sqrt{{\delta_0}/{k}}, \quad & k>\kappa_T,\\
  \sqrt{\delta_0/\kappa_T}=e^{r_0 T}, \quad & k \leq
  \kappa_T.\end{cases}
$$
Obviously $\partial \widehat{\Lambda}$ is square integrable (and hence
the problem is well-posed) if and only if $T<\infty$. In that case, we
compute
$$
\widehat{Y}_t \set\partial \widehat{\Lambda}_{\widetilde{\kappa}_t} = \begin{cases} \frac{1}{2} \sqrt{\delta_0/\kappa_t} =
  \frac{1}{2} e^{r_0t}, \quad &t<T,\\ 
\sqrt{\delta_0/\kappa_T}=e^{r_0T},
  \quad &t \geq T,\end{cases}
$$
and for any $y>0$ the order schedule from~\eqref{eq:10},
$$
X^y_t \set \delta_0
\left(\frac{1}{2}y - \eta_0\right)^+ + \frac{1}{2}y \delta_0 r_0 \left(t \wedge T-\tau^y\right)^+ +
\frac{1}{2}\delta_0 y 1_{[T,\infty]}(t), \quad t \geq 0,
$$
with
$$
\tau^y \set \begin{cases} \left(\frac{1}{r_0} \log \frac{2 \eta_0}{y}\right)^+ \wedge
T, & y > 2 \eta_0 e^{-r_0 T},\\
T, &\eta_0 e^{-r_0 T} \leq y \leq 2 \eta_0 e^{-r_0 T},\\
\infty & y<\eta_0 e^{-r_0 T},
\end{cases}
$$
is optimal for the total volume it trades. In particular, if
$\eta_0=0$, we find that
$$
X^y_t =  \frac{y \delta_0}{2} \left(1+ r_0 (t  \wedge
T) + 1_{[T,\infty]}(t)\right), \quad t \geq 0.
$$
So choosing $y^* \set x/(\delta_0(1+r_0 T/2))$ yields $X^* = X^{y^*}$
with $X^*_\infty = x$. We therefore recover the result of
\citet{ObizhaevaWang}: If $\eta_0=0$, i.e., if there have been no
previous orders, it is optimal to place orders of size $y^*
\delta_0/2$ at both $t=0$ and $t=T$, and to place orders at the
constant rate $y^* \delta_0 r_0/2$ in between; cf. Figure~\ref{fig:1}.

\begin{figure}[htbp]
		\centering
			\includegraphics[width=.8\linewidth]{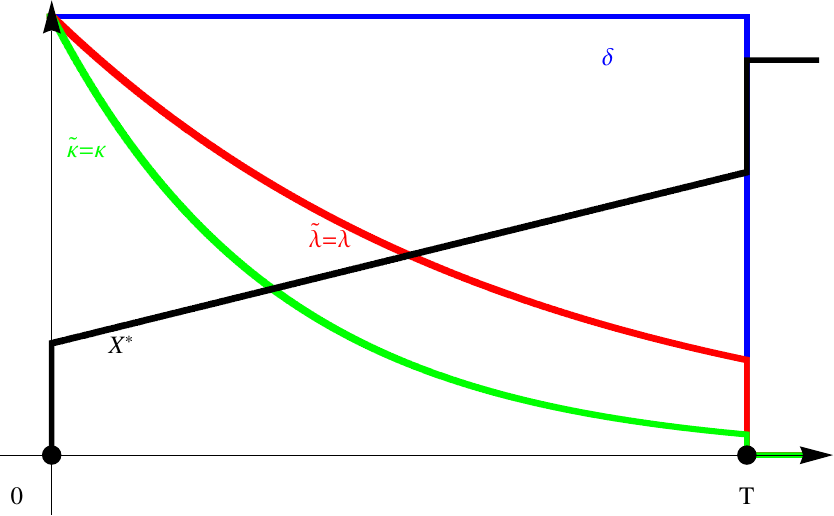}
 \caption{Optimal order schedule $X^*$ (black) for constant market
   depth $\delta$
   (blue), its resilience adjustment $\lambda=\widetilde{\lambda}$
   (red), $\kappa=\widetilde{\kappa}$ (green) over a finite horizon~$T$.}\label{fig:1}
\end{figure}

\subsection{Time-varying market depth}\label{sec:Numerics}

We next illustrate that the above order placement strategy of
\cite{ObizhaevaWang} is indeed strongly dependent on constant market
depth and resilience. Figure~\ref{fig:01} below exhibits how a
fluctuating market depth affects the timing of the optimal order
placement as provided by Corollary~\ref{cor:1}. Note that we include a
shut-down period for the market over the time period $(t_0,t_1)$ when
market depth vanishes. The corresponding concepts introduced by
Theorem~\ref{thm:2} are illustrated in Figure~\ref{fig:0} below.

\begin{figure}[htbp]
		\centering
			\includegraphics[width=.8\linewidth]{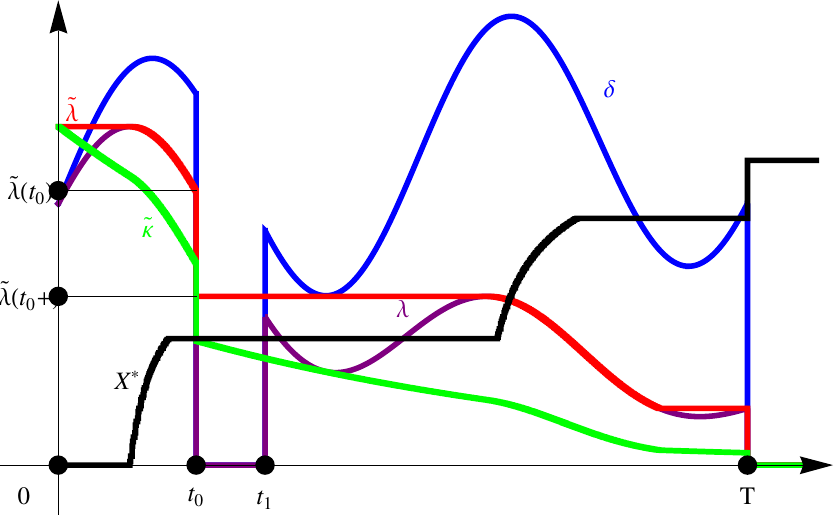}
 \caption{A specification of market depth $\delta$ (blue) with finite
   horizon $T$, its resilience adjustment $\lambda$ (purple), the
   corresponding decreasing envelope $\widetilde{\lambda}$ (red) and
   $\widetilde{\kappa}$ (green) along with an optimal order schedule
   $X^*$ (black).}\label{fig:01}
\end{figure}

\begin{figure}[htbp]
		\centering
			\includegraphics[width=.8\linewidth]{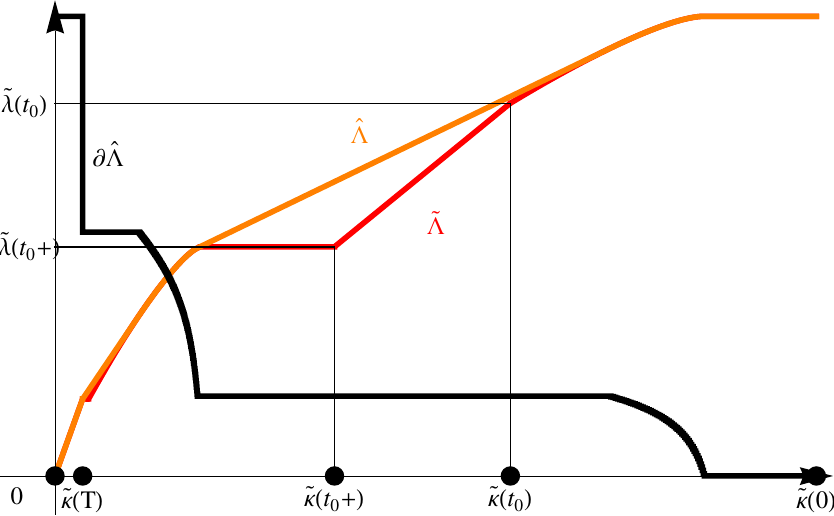}
 \caption{The decreasing envelope of resilience adjusted market depth   $\widetilde{\Lambda}$ (red),
its concave envelope $\widehat{\Lambda}$ (orange) and the density
$\partial \widehat{\Lambda}$ (black).}\label{fig:0}
\end{figure}

If we decrease the resilience parameter to $r_0=0$, i.e., we assume
permament price impact of the broker's orders, the focus on peaks of
market depth sharpens to the extent that eventually only one huge
order is placed when market depth reaches its global maximum; see
Figure~\ref{fig:3}.

\begin{figure}[htbp]
		\centering
			\includegraphics[width=.8\linewidth]{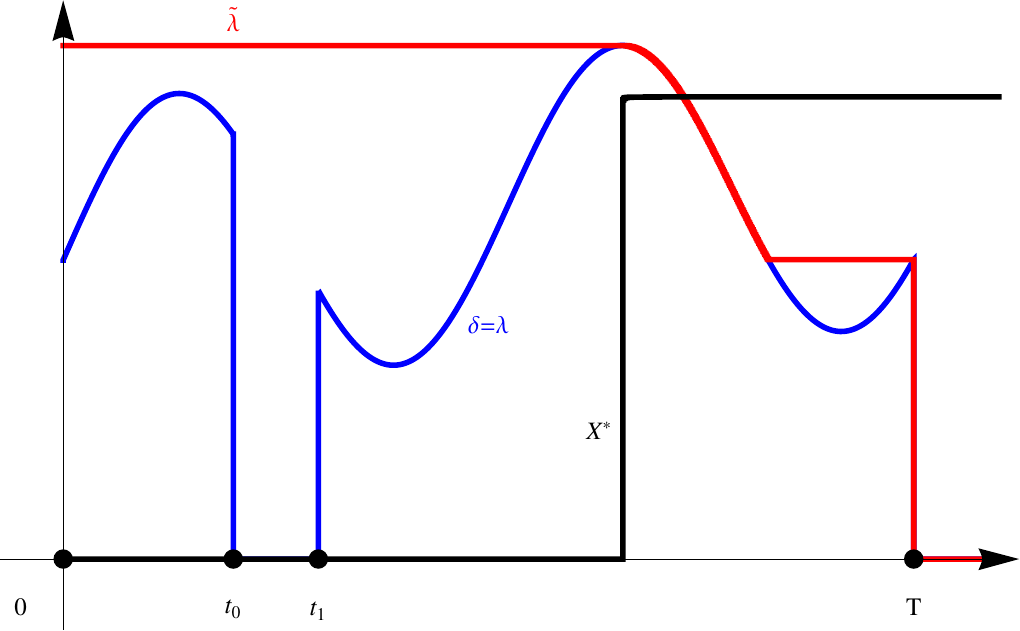}
	 \caption{Optimal order schedule $X^*$ (black) without market
           resilience and time-varying market depth $\delta$ (blue).}\label{fig:3}
\end{figure}

\begin{Proposition}\label{pro:3}
  If $r \equiv 0$ and $\delta$ satisfies Assumption~\ref{as:2}, the
  solutions to optimization problem~\eqref{eq:3} are precisely those
  order schedules $X^* \in \mathcal{X}$ with $\{dX^*>0\} \subset
  \argmax \delta$.
\end{Proposition}
\begin{proof}
  When $r \equiv 0$, $\rho \equiv 1$ and so $\eta^X_t =
  \eta_0+\int_{[0,t]} \frac{dX_s}{\delta_s} \geq
  \eta_0+\frac{X_t}{\max \delta}$, $t \geq 0$. Thus,
  $$
  C(X) \geq \eta_0 x+\frac{x^2}{2 \max \delta}, \quad X \in
  \mathcal{X},
  $$
  with equality for all $X^* \in \mathcal{X}$ with $\{dX^*>0\} \subset
  \argmax \delta$.
\end{proof}

Conversely, with high resilience, orders tend to be spread out more
around local maxima of market depth as illustrated by
Figure~\ref{fig:4}. Figures~\ref{fig:01} and~\ref{fig:4} also show
that the precise moments when it is optimal to issue orders would be
hard to guess in advance. Hence, an approach via classical calculus of
variations as in \citet{FruthUrusovSchoeneborn} or via the methods of
\citet{AcevedoAlfonsi} seems infeasible in these general cases.

\begin{figure}[htbp]
		\centering
			\includegraphics[width=.8\linewidth]{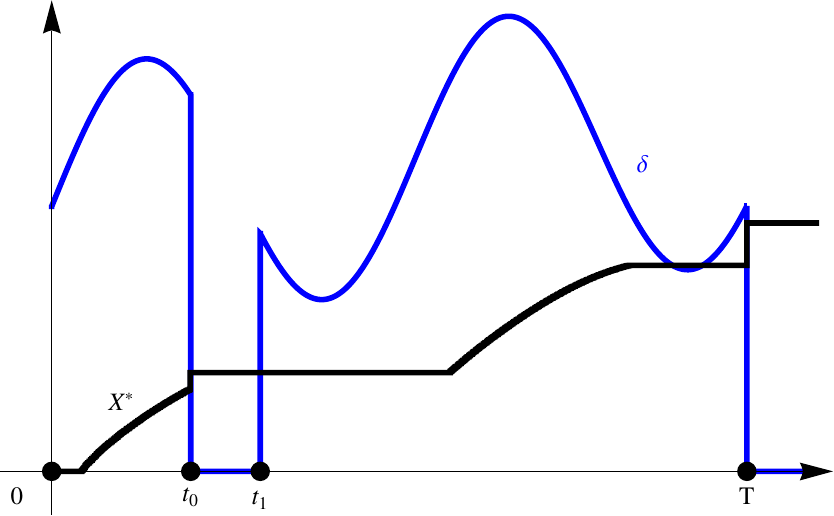}
		 \caption{Optimal order schedule $X^*$ (black) with
                   strong market resilience for time-varying market
                   depth $\delta$
   (blue).}\label{fig:4}
\end{figure}

\section{Proofs}
\label{sec:proofs}

We first prove that the original problem~\eqref{eq:3} can indeed be
reformulated as~\eqref{eq:7} by giving the 

\begin{proofof}{Proposition~\ref{pro:0}}
  We first observe that for $X \in \mathcal{X}$ the mapping
  in~\eqref{eq:6} defines an increasing right-continuous $Y$ with $Y =
  \rho \eta^X$. Because $C(X)<\infty$, $\eta^X$ is $dX$-integrable and
  thus finite on $\{X<x\}$. Hence, $Y$ is finite on this set as well
  and we conclude $dX = \lambda \,dY$. It follows by elementary
  calculus that $K(Y)=C(X)$ and, thus, $Y \in \mathcal{Y}$ as desired.

  Conversely, for $Y \in \mathcal{Y}$, $\kappa=\lambda/\rho$ is
  $d(Y^2)$-integrable. Since $\rho>0$ is continuous this implies that
  $\lambda$ is locally $dY$-integrable and so $X$ given
  by~\eqref{eq:6} is right-continuous and increasing with $dX =
  \lambda \,dY$. By the same reasoning as above this implies
  $C(X)=K(Y)$ as well as $X \in \mathcal{X}$.
\end{proofof}

We next characterize when problem~\eqref{eq:7} is convex:

\begin{proofof}{Proposition~\ref{pro:1}}
  If $\kappa$ is upper semi-continuous and decreasing, it is also
  left-continuous and we can use Fubini's theorem to write
 \begin{displaymath}
  K(Y) = \frac{1}{2}\left(\kappa_\infty (Y^2_\infty-\eta_0^2) - \int_{[0,\infty)} (Y^2_t-\eta_0^2)
    \,d\kappa_t\right)  
 \end{displaymath}
 for any right-continuous increasing $Y$ with $Y_{0-} =\eta_0$. Hence,
 $K=K(Y)$ is obviously convex in such $Y$ with strict convexity
 holding true on its domain for strictly decreasing $\kappa$.

 Conversely, consider for $0 \leq s < t <\infty$ the function $Y \set
 \eta_0 + a 1_{[s,\infty]}+b 1_{[t,\infty]}$. Then
\begin{align*}
  K(Y) &= \frac{1}{2}\left(\kappa_s ((a+\eta_0)^2-\eta_0^2)+\kappa_t
    \left((a+b+\eta_0)^2-(a+\eta_0^2)\right)\right)\\ 
&= \frac{1}{2} \kappa_s
  a^2 +\kappa_t ab+ \frac{1}{2} \kappa_t b^2 + \eta_0(a \kappa_s + b \kappa_t)
 \end{align*}
 is convex in $a,b>0$ if and only if $\kappa_s \geq \kappa_t \geq 0$,
 with strict inequalities corresponding to strict convexity.
\end{proofof}

In order to prepare the proof of Theorem~\ref{thm:1} let us recall
that for any increasing $Z:[0,\infty) \to \mathbf{R}$ we let
\begin{displaymath}
  \{dZ>0\} \set \descr{t \geq 0}{Z_{t-} < Z_u \ttext{for all} u > t}
\end{displaymath}
denote the collection of all points of increase towards the right. For
a decreasing $Z$ we let $\{dZ<0\}\set\{d(-Z)>0\}$. In either case we let
$\supp dZ$ denote the support of the measure $dZ$, i.e., the smallest
closed set whose complement has vanishing $dZ$-measure.

\begin{Lemma}
  \label{lem:1}
  For upper-semicontinuous, bounded $\lambda:[0,\infty) \to \mathbf{R}$, we have that
  $\widetilde{\lambda}_t \set \sup_{u \geq t} \lambda_u$ is
 left-continuous and decreasing with
  \begin{equation}
    \label{eq:11}
    \{d\widetilde{\lambda}<0\} \subset \{\widetilde{\lambda}=\lambda\}\,.
  \end{equation}
  Moreover, we have the partition
  \begin{equation}
    \label{eq:12}
   \mathbf{R} = \{d\widetilde{\lambda}<0\} \cup \bigcup_{n \in
     N_1} [l_n,r_n) \cup \bigcup_{n \in N_2} (l_n,r_n)
  \end{equation}
  where $(l_n,r_n)$, $n \in N$, are the disjoint open
  intervals forming $\mathbf{R}\backslash \supp d\widetilde{\lambda}$
  and where $N_1 = \descr{n \in N}{l_n\geq0, \Delta_{l_n}
    \widetilde{\lambda}=0}$ and $N_2 = N \backslash N_1$.
\end{Lemma}
\begin{proof}
 Left-continuity of $\widetilde{\lambda}$ and relation~\eqref{eq:11}
  are immediate. Note next that
  $\{d\widetilde{\lambda}<0\} \subset \supp d\widetilde{\lambda}$
  and therefore $\mathbf{R} \backslash \{d\widetilde{\lambda}<0\}  \supset
  \bigcup_{n \in N} (l_n,r_n)$. Hence, to deduce
  partition~\eqref{eq:12} it suffices to observe that for $n \in N_1$
  we have $l_n \not\in \{d\widetilde{\lambda}<0\}$ and that for $t
  \geq 0$ such that $\widetilde{\lambda}_t = \widetilde{\lambda}_u$ for
  some $u > t$ we have $(t,u) \subset (l_n,r_n)$ for some $n \in
  N$, and thus $t \in (l_n,r_n)$ or $t=l_n$ with $\Delta_{l_n} \widetilde{\lambda}=0$.
\end{proof}

The main tool in the proof of Theorem~\ref{thm:1} is the following

\begin{Lemma}
 \label{lem:2}
  Under the conditions of Theorem~\ref{thm:1}, we can find for any
  increasing, right-continuous $Y \geq \eta_0$ an increasing,
  right-continuous $\widetilde{Y} \geq \eta_0$
  such that $\widetilde{Y} \leq Y$ and
  \begin{itemize}
  \item[(i)] $\int_{[0,\infty)} \lambda_t \,dY_t = \int_{[0,\infty)} \lambda_t
    \,d\widetilde{Y}_t$,
\item[(ii)] $\{d\widetilde{Y}>0\} \subset \{d\widetilde{\lambda}<0\}$,
\item[(iii)] $K(Y) \geq K(\widetilde{Y})=\widetilde{K}(\widetilde{Y})$.
  \end{itemize}
\end{Lemma}
\begin{proof}
 We let $I_n$, $n \in N$, denote the disjoint intervals of
 Lemma~\ref{lem:1} forming the complement of
 $\{d\widetilde{\lambda}<0\}$ and we will use $l_n$, $r_n$ to denote their
 respective boundaries. For the one interval $I_n$ whose left bound is $l_n=-\infty$
 we now redefine, for simplicity of notation, 
 $l_n \set 0$ provided that $r_n>0$; if, by contrast, this $I_n$ is just the
 negative half line we can and shall remove it from consideration in
 the sequel. Similarly, if $r_n = \infty$ for some $n \in N$, it
 follows from
 Assumption~\ref{as:2} that $\delta_t = \lambda_t = \kappa_t \equiv 0$ on $I_n$ which thus can be
 disregarded as well.

 Observe then that
 \begin{equation}
   \label{eq:13}
   \sup_{I_n} \lambda
 = \lambda_{r_n}\,,
 \end{equation}
by upper semi-continuity of $\lambda$ and our choice when to include
$l_n$ in $I_n$ and when not.

Let, for $t \geq 0$, 
 \begin{align*}
   \widetilde{Y}_t \set \eta_0&+\int_{[0,t]} 1_{\{d\widetilde{\lambda}<0\}}(s)
   \,dY_s  +
\sum_{n \in N, r_n  \leq t} \int_{I_n}
\frac{\lambda_s}{\lambda_{r_n}} \,dY_s\,.
 \end{align*}
We first note that $\widetilde{Y} \leq Y$. Indeed
\begin{align*}
  Y_t - \widetilde{Y}_t = &\int_{[0,t]}
  1_{\mathbf{R}\backslash\{d\widetilde{\lambda}<0\}}(s)
  \,(dY_s-d\widetilde{Y}_s)\\
=&\sum_{n \in N, l_n \leq t} \left(\int_{I_n \cap [0,t]} \,dY_s
  - 1_{[r_n, \infty)}(t) \int_{I_n}\frac{\lambda_s}{\lambda_{r_n}}
  \,dY_s\right) 
\end{align*}
is nonnegative because of~\eqref{eq:13}.

Assertion~(i) is readily checked using the partition given
by~\eqref{eq:12}. For assertion~(ii) it suffices to observe that all
$r_n$, $n \in N$, are contained in $\{d\widetilde{\lambda}<0\}$. 

In order to prove assertion~(iii), we first note that
$K(\widetilde{Y})=\widetilde{K}(\widetilde{Y})$ is an immediate
consequence of~(ii) and~\eqref{eq:11}. To establish $K(Y) -
K(\widetilde{Y}) \geq 0$ we decompose this difference into its
contributions from the different parts in the partition given by~\eqref{eq:12},
each of which will be shown to be nonnegative.

From $\{d\widetilde{\lambda}<0\}\backslash\descr{r_n}{n \in N}$ we
collect
\begin{align*}
  \frac{1}{2} &\int_{[0,\infty)  \cap
    (\{d\widetilde{\lambda}<0\}\backslash\descr{r_n}{n \in N})}
  \kappa_t \, \left[d(Y^2_t)-d(\widetilde{Y}^2_t)\right] \\
&=\int_{[0,\infty)  \cap
    (\{d\widetilde{\lambda}<0\}\backslash\descr{r_n}{n \in N})}
  \kappa_t \, \left[\left(Y_{t-} + \frac{1}{2} \Delta_t
      Y\right)\,dY_t-\left(\widetilde{Y}_{t-} + \frac{1}{2} \Delta_t
      \widetilde{Y}\right)\,d \widetilde{Y}_t \right]
\end{align*}
which is nonnegative because $Y \geq \widetilde{Y}$ and because $dY_t
= d\widetilde{Y}_t$ for $t \in
\{d\widetilde{\lambda}<0\}\backslash\descr{r_n}{n \in N}$ by
construction.

From $I_n \cup \{r_n\}$, $n \in N$, we get the contribution
\begin{align*}
 \frac{1}{2} \left\{ \int_{I_n \cup \{r_n\}}\kappa_t d(Y^2_t) -
 \kappa_{r_n}\left[\left(\widetilde{Y}_{r_n-}+
     \int_{I_n \cup \{r_n\}}\frac{\lambda_s}{\lambda_{r_n}}
     \,dY_s\right)^2-\widetilde{Y}_{r_n-}^2\right]\right\}
\end{align*}
for which we note that its $[\dots]$-part can be written as
\begin{align*}
  \frac{1}{2}&\left[\left(\widetilde{Y}_{r_n-}+
     \int_{I_n \cup \{r_n\}}\frac{\lambda_s}{\lambda_{r_n}}
     \,dY_s\right)^2-\widetilde{Y}_{r_n-}^2\right]\\
=&\frac{1}{2} \left( \int_{I_n \cup \{r_n\}}\frac{\lambda_s}{\lambda_{r_n}}
     \,dY_s\right)^2+ \widetilde{Y}_{r_n-}\int_{I_n \cup \{r_n\}}\frac{\lambda_s}{\lambda_{r_n}}
     \,dY_s\\
= &\int_{I_n \cup \{r_n\}} \int_{(I_n \cup \{r_n\}) \cap [l_n,t)} \frac{\lambda_s}{\lambda_{r_n}} \frac{\lambda_t}{\lambda_{r_n}}
     \,dY_s \,dY_t+\frac{1}{2} \sum_{\Delta_t Y\not=0, t \in I_n \cup \{r_n\}}
     \left(\frac{\lambda_t}{\lambda_{r_n}}\right)^2 (\Delta_t Y)^2\\
&+ \widetilde{Y}_{r_n-}\int_{I_n \cup \{r_n\}}\frac{\lambda_s}{\lambda_{r_n}}
     \,dY_s\,.
\end{align*}
Hence, using~\eqref{eq:13} again, we obtain with $y_n \set Y_{l_n-}$ if
$l_n \in I_n$ and $y_n \set Y_{l_n}$ otherwise that
\begin{align*}
 \frac{1}{2}[\dots] 
 \leq & \int_{I_n \cup \{r_n\}} (Y_{t-}-y_n)\frac{\lambda_t}{\lambda_{r_n}}
     \,dY_t+\frac{1}{2} \sum_{\Delta_t Y\not=0, t \in I_n \cup \{r_n\}}
     \frac{\lambda_t}{\lambda_{r_n}} (\Delta_t Y)^2\\
&+ \widetilde{Y}_{r_n-}\int_{I_n \cup \{r_n\}}\frac{\lambda_s}{\lambda_{r_n}}
     \,dY_s\\
\leq & \int_{I_n \cup \{r_n\}} (Y_{t-}-y_n)\frac{\lambda_t}{\lambda_{r_n}}
     \,dY_t+\frac{1}{2} \sum_{\Delta_t Y\not=0, t \in I_n \cup \{r_n\}}
     \frac{\lambda_t}{\lambda_{r_n}} (\Delta_t Y)^2\\
&+ y_n\int_{I_n \cup \{r_n\}}\frac{\lambda_s}{\lambda_{r_n}}
     \,dY_s\\
= &\frac{1}{2}\int_{I_n \cup \{r_n\}} \frac{\lambda_t}{\lambda_{r_n}}
     \,d(Y^2_t)
\end{align*}
where the second estimate holds since
$\widetilde{Y}_{r_n-} = \widetilde{Y}_{l_n}\leq
y_n$ because of~(ii). Since $\rho = \lambda/\kappa$ is
increasing by assumption, we have 
\begin{displaymath}
  \frac{\lambda_t}{\lambda_{r_n}} = \frac{\rho_t}{\rho_{r_n}}
  \frac{\kappa_t}{\kappa_{r_n}} \leq   \frac{\kappa_t}{\kappa_{r_n}} 
\end{displaymath}
and thus
\begin{displaymath}
  \frac{1}{2}\kappa_{r_n}[\dots] \leq  \frac{1}{2}\int_{I_n \cup \{r_n\}} \kappa_t
     \,d(Y^2_t)
\end{displaymath}
as remained to be shown.
\end{proof}

With the preceding policy improvement lemma it is now easy to give the

\begin{proofof}{Theorem~\ref{thm:1}}
  By Lemma~\ref{lem:2} and using its notation, we can find for any $Y
  \in \mathcal{Y}$ a $\widetilde{Y} \in \widetilde{\mathcal{Y}} \cap
  \mathcal{Y}$ such that
 \begin{displaymath}
  \widetilde{K}(Y) \geq K(Y) \geq K(\widetilde{Y}) = \widetilde{K}(\widetilde{Y})\,.
 \end{displaymath}
 As a result, $\inf_{{\mathcal{Y}}}{K} =
 \inf_{\widetilde{\mathcal{Y}}} \widetilde{K}$. Moreover, if
 $\widetilde{Y}^* \in \widetilde{\mathcal{Y}}$ attains the latter
 infimum we can apply Lemma~\ref{lem:2} to $\widetilde{\lambda}$ and
 $\widetilde{K}$ instead of $\lambda$ and $K$ to obtain another
 optimal $\widetilde{Y}^{**} \in \widetilde{\mathcal{Y}}$ which
 satisfies in addition $\{d\widetilde{Y}^{**} >0\}\subset
 \{d\widetilde{\lambda}<0\}$. By Lemma~\ref{lem:1}, the latter set is
 contained in $ \{\lambda =
 \widetilde{\lambda}\}=\{\kappa=\widetilde{\kappa}\}$ and thus this
 $\widetilde{Y}^{**}$ is also contained in $\mathcal{Y}$ and optimal
 for~\eqref{eq:7} as well.
\end{proofof}

Let us next derive the first-order conditions of the convexified
problem~\eqref{eq:8} in the

\begin{proofof}{Proposition~\ref{pro:2}}
Recalling that $\widetilde{\kappa}_\infty=0$, we obtain by Fubini's theorem 
\begin{equation}
  \label{eq:14}
  \widetilde{K}(Y) = -\frac{1}{2} \int_{[0,\infty)} (Y^2_t- \eta_0^2) \,
  d\widetilde{\kappa}_t\,.
\end{equation}

For necessity, we observe that for any $Y \in
\widetilde{\mathcal{Y}}$ and $0 < \epsilon \leq 1$ we have
\begin{align*}
 0 \leq & \widetilde{K}(\epsilon Y+(1-\epsilon) Y^*) - \widetilde{K}(Y^*) \\ = &-\epsilon
 \int_{[0,\infty)} (Y_t-Y^*_t)Y^*_t\, d\widetilde{\kappa}_t -
 \frac{\epsilon^2}{2}  \int_{[0,\infty)} (Y_t-Y^*_t)^2\,d\widetilde{\kappa}_t  
\end{align*}
which, upon division by $\epsilon >0$ and letting $\epsilon \downarrow
0$, yields that $Y^*$ also solves the linear problem
\begin{equation}
  \label{eq:15}
  \text{Minimize } \; -\int_{[0,\infty)} Y^*_t Y_t \,d\widetilde{\kappa}_t 
  \mtext{subject to} Y \in \widetilde{\mathcal{Y}}\,.
\end{equation}
Equivalently, due to Fubini's theorem, $Y^*$ is a solution to the problem:
\begin{equation}
  \label{eq:16}
  \text{Minimize } \; \int_{[0,\infty)}
 \left( -\int_{[t,\infty)} Y^*_{u} \,d\widetilde{\kappa}_u\right) \,dY_t
  \mtext{subject to} Y \in \widetilde{\mathcal{Y}}\,.
\end{equation}
As a consequence, $Y^*$ can solve~\eqref{eq:15} only if $dY^*_t>0$
exclusively at those times $t \geq 0$ when $-\int_{[t,\infty)} Y^*_{u}
\,d\widetilde{\kappa}_u/\widetilde{\lambda}_t$ attains its infimum
over $\{\widetilde{\lambda}>0\}$. Hence, this infimum is actually a
minimum and is thus strictly positive. Denoting it by $y>0$ shows the
necessity of~\eqref{eq:9}.

For sufficiency we use~\eqref{eq:14} again to deduce that for $Y \in \widetilde{\mathcal{Y}}$:
\begin{align*}
  \widetilde{K}(Y)-\widetilde{K}(Y^*) &= -\frac{1}{2}\int_{[0,\infty)}
  ((Y_t)^2-(Y^*_t)^2) \,d\widetilde{\kappa}_t %\\&
\geq -\int_{[0,\infty)}
  Y^*_t(Y_t-Y^*_t) \,d\widetilde{\kappa}_t\,.
\end{align*}
The last term is nonnegative if $Y^*$ solves~\eqref{eq:15}, which due
to the equivalence of~\eqref{eq:15} and~\eqref{eq:16} amounts to our
first-order condition~\eqref{eq:9}.
\end{proofof}

The construction of solutions to the first order conditions given in Theorem~\ref{thm:2} can now be established:

\begin{proofof}{Theorem~\ref{thm:2}}
  $\widetilde{\Lambda}$ is continuous on $[0,\widetilde{\kappa}_0]$
  since so is $k \mapsto\tau_k$ because of the strict monotonicity of
  $\rho$ and, thus, of $\widetilde{\kappa}$ on
  $\{\widetilde{\kappa}>0\}$. $\widetilde{\Lambda}$ is increasing
  because, along with $\widetilde{\kappa}_t$, also
  $\widetilde{\Lambda}_{\widetilde{\kappa}_t}=\widetilde{\kappa}_t
  \rho_t=\widetilde{\lambda}_t$ is decreasing in $t \geq 0$.  Absolute
  continuity of the concave envelope $\widehat{\Lambda}$ follows from
  the continuity of $\widetilde{\Lambda}$.

  The monotonicity of $\widetilde{Y}^*$ is obvious from the
  monotonicity of $\widetilde{\kappa}$ and $\partial
  \widehat{\Lambda}$. For its right-continuity note that $\lim_{t
    \downarrow t_0} \widetilde{Y}^*_t = (y\partial
  \widehat{\Lambda}_{\widetilde{\kappa}_{t_0+}}) \vee \eta_0$ by
  left-continuity of $\partial \widehat{\Lambda}$ and its definition
  at $0$. Hence, our assertion amounts to $\partial
  \widehat{\Lambda}_{k_0}= \partial \widehat{\Lambda}_{k_1}$ where
  $k_0 \set \widetilde{\kappa}_{t_0+}$ and $k_1 \set
  \widetilde{\kappa}_{t_0} \geq k_0$. If $k_0 = k_1$ there is nothing
  to show. In case $k_0 < k_1$, $\tau_k = \tau_{k_1}$ for $k \in
  [k_0,k_1)$ and, thus, $\widetilde{\Lambda}$ is linear with slope
  $\rho_{\tau_{k_0}}$ on this interval.  As a consequence,
  $\widehat{\Lambda}$ is linear there as well and, thus, $\partial
  \widehat{\Lambda}_{k_1}=\partial \widehat{\Lambda}_{k_0+}$ by
  left-continuity of $\partial \widehat{\Lambda}$. Hence, it suffices
  to show that there is no downward jump in $\partial
  \widehat{\Lambda}$ at $k_0$.  If there was such a jump then, by the
  properties of concave envelopes, necessarily
  $\widehat{\Lambda}_{k_0}=\widetilde{\Lambda}_{k_0}$ and $\partial
  \widehat{\Lambda}_{k_0+} \geq \rho_{\tau_{k_0}}$.  Hence, for $k
  \leq k_0$ we would have
  \begin{align*}
k \rho_{\tau_{k_0}} &\leq \widetilde{\Lambda}_k \leq 
\widehat{\Lambda}_k \leq \widehat{\Lambda}_{k_0} + \partial
\widehat{\Lambda}_{k_0+}(k-k_0) \leq k \rho_{\tau_{k_0}}\,,
  \end{align*}
  where the first estimate is due to the monotonicity of $\rho$, the
  second is the envelope property of $\widehat{\Lambda}$, the third
  follows from its concavity and the last is a consequence of the just
  derived properties of $\widehat{\Lambda}$ and $\partial
  \widehat{\Lambda}$ at $k_0$.  We would thus have equality everywhere
  in the above estimates and in particular $\partial
  \widehat{\Lambda}_{k_0} =\rho_{\tau_{k_0}} \leq \partial
  \widehat{\Lambda}_{k_0+}$. This is a contradiction to the presumed
  downward jump of $\partial \widehat{\Lambda}$ at $k_0$.

To verify that $\widetilde{Y}^*$ satisfies the first oder
condition~\eqref{eq:9}, let us first argue that
\begin{align*}
-\int_{[t,\infty)} \widetilde{Y}^*_u
\,d\widetilde{\kappa}_u & \geq - y \int_{[t,\infty)} \partial
\widehat{\Lambda}_{\widetilde{\kappa}_u} \,
d\widetilde{\kappa}_u 
= y \int_0^{\widetilde{\kappa}_t} \partial \widehat{\Lambda}_k \,dk =
y \widehat{\Lambda}_{\widetilde{\kappa}_t}\\
& \geq y \widetilde{\Lambda}_{\widetilde{\kappa}_t} = y
\widetilde{\kappa}_t
\rho_t%{\tau_{\widetilde{\kappa}_t}}
=y \widetilde{\lambda}_t\,.
\end{align*}
Indeed, the first estimate is immediate from the definition of
$\widetilde{Y}^*$. The first identity follows by the change-of-time
formula for Lebesgue-Stieltjes-integrals: just observe that
$\widetilde{\kappa}_\infty =0$ by Assumption~\ref{as:2} and that
$\partial \widehat{\Lambda}$ is constant on those intervals contained
in $[0,\widetilde{\kappa}_0]$ which $\widetilde{\kappa}$ jumps across
because $\widetilde{\Lambda}$ is linear on such intervals. The second
identity follows from the absolute continuity of $\widehat{\Lambda}$
and because $\widehat{\Lambda}_0 = \widetilde{\Lambda}_0=0$, again by
Assumption~\ref{as:2}. The second estimate holds because
$\widehat{\Lambda} \geq \widetilde{\Lambda}$ by definition of concave
envelopes and for the last identity we note that
$\tau_{\widetilde{\kappa}_t}=t$ if $\widetilde{\kappa}_t>0$ and
$\widetilde{\kappa}_t=\widetilde{\lambda}_t=0$ otherwise. Finally, we
observe that $d\widetilde{Y}^*_t>0$ can only happen when $y \partial
\widehat{\Lambda}_{\widetilde{\kappa}}$ has increased above $\eta_0$
which ensures equality in the first of the above estimates. Equality
in the second holds for such $t$ as well because if $\partial
\widehat{\Lambda}_{\widetilde{\kappa}}$ increases at time $t$,
$\partial \widehat{\Lambda}$ must decrease at
${\widetilde{\kappa}}_t$, and so $\widetilde{\Lambda}$ coincides with
its concave envelope $\widehat{\Lambda}$ at this point.
\end{proofof}

We are now in a position to wrap up and give the

\begin{proofof}{Corollary~\ref{cor:1}}
Let $\widehat{Y}_t \set \partial
\widehat{\Lambda}_{\widetilde{\kappa}_t}$, $\widehat{Y}_{0-}\set0$ and
define $Y^y_t \set (y\widehat{Y}_t) \vee \eta_0$, $t \geq 0$, $Y^y_{0-} \set\eta_0$.

As a first step we check that
\begin{equation}
  \label{eq:177}
  d\widehat{Y}_{t_0}>0 \text{ only at times $t_0 \geq 0$ when }
  \widetilde{\lambda}_{t_0} = \lambda_{t_0}\,.
\end{equation}
In fact, we will show that $d\widetilde{\lambda}_{t_0}<0$ for such
$t_0$. If $\Delta_{t_0} \widetilde{\kappa}<0$, this is obvious. So let
us suppose that $\widetilde{\kappa}_{t_0+}=\widetilde{\kappa}_{t_0}$
and assume that there is $t_1>t_0$ such that $\widetilde{\lambda}_{t}
= \widetilde{\lambda}_{t_0}$ for $t \in [t_0,t_1]$. In that case,
$\widetilde{\Lambda}$ is constant on the interval
$(\widetilde{\kappa}_{t_1}, \widetilde{\kappa}_{t_0}]$.  Because
$d\widehat{Y}_{t_0}>0$, the density $\partial \widehat{\Lambda}$ must
decrease at $k_0\set
\widetilde{\kappa}_{t_0+}=\widetilde{\kappa}_{t_0}$ and so the
envelope $\widehat{\Lambda}$ coincides with $\widetilde{\Lambda}$ at
this point. Concavity and monotonicity of $\widehat{\Lambda}$ then
imply, however, that $\partial \widehat{\Lambda} = 0$ around $k_0$, a
contradiction to its decrease there.

Let us next prove that $|\partial
\widehat{\Lambda}_k|_{\mathbf{L}^2}<\infty$ if and only if
$\widetilde{\lambda}$ is $d\widehat{Y}$-integrable. To see this we
argue that with $\partial \widehat{\Lambda}_{\widetilde{\kappa}_{0-}}
\set 0$ we have
\begin{align*}
  \int_{[0,\infty)} {\lambda}_t \,d\widehat{Y}_t &= \int_{[0,\infty)}
  \widetilde{\Lambda}_{\widetilde{\kappa}_t} \,d (\partial
  \widehat{\Lambda}_{\widetilde{\kappa}_t}) =  \int_{[0,\infty)}
  \widehat{\Lambda}_{\widetilde{\kappa}_t} \,d (\partial
  \widehat{\Lambda}_{\widetilde{\kappa}_t})\\
&=\int_0^{\widetilde{\kappa}_0} \partial
\widehat{\Lambda}_l %\left(
\partial\widehat{\Lambda}_{\widetilde{\kappa}_{\tau_l}}
%-\partial\widehat{\Lambda}_{\widetilde{\kappa}_{0}}\right) 
\,dl
=\int_0^{\widetilde{\kappa}_0}(\partial
\widehat{\Lambda}_l)^2 \,dl
%-\widehat{\Lambda}_{\widetilde{\kappa}_{0}} \partial
%\widehat{\Lambda}_{\widetilde{\kappa}_{0}}
\,.
\end{align*}
Indeed, the first identity is just~\eqref{eq:177} and the definition
of $\widehat{Y}$ and $\widetilde{\Lambda}$. The second identity holds
because $\widetilde{\Lambda}=\widehat{\Lambda}$ at points where
$\partial \widehat{\Lambda}$ changes; the third identity follows from
an application of Fubini's theorem after writing
$\widehat{\Lambda}_{\widetilde{\kappa}_t} =
\int_0^{\widetilde{\kappa}_t} \partial \widehat{\Lambda}_l \,dl$ and
the last equality holds since $\partial \widehat{\Lambda}$ is
left-continuous and constant over intervals that $\widetilde{\kappa}$
jumps across.

So if $|\partial \widehat{\Lambda}|_{\mathbf{L}^2}<\infty$, then
$X^y_t \set \int_{[0,t]} \lambda \,dY^y$, $t \geq 0$, is real-valued,
right-continuous and increasing in $t$. Moreover, $X^y_\infty$ is
increasing in $y \geq 0$ with $X^0_\infty = 0$ and $X^y_\infty \geq
X^y_0 \to \infty$ as $y \uparrow \infty$. In fact, $X^y_\infty = y
\int_{[0,\infty)} \lambda \,d\widehat{Y}$ for $y \geq
\eta_0/\widehat{Y}_0$ (where $0/0 \set \infty$) and, for $y \in
[\eta_0/\widehat{Y}_\infty,\eta_0/\widehat{Y}_0]$, $X^y_\infty = y
\int_{[\tau_y,\infty)} \lambda \,d\widehat{Y} =
\int_{[\tau_{y+},\infty)} \lambda \,d\widehat{Y}$ where $\tau_y \set
\inf\descr{t \geq 0}{y>\eta_0/\widehat{Y}_t}$. Hence, $X^y_\infty$ is
in fact continuous and strictly increasing from 0 to $\infty$ in $y
\geq \eta_0/\widehat{Y}_\infty$ and we thus obtain existence and
uniqueness of $y^*>0$ with $X^{y^*}_\infty = x$.  Hence, we can
conclude that $X^* \set X^{y^*}$ is contained in $\mathcal{X}$ (and
that thus the corresponding $Y^* = Y^{y^*}$ of~\eqref{eq:6} is
contained in $\mathcal{Y}$) once we have established that
$K(Y^*)<\infty$. For this it suffices to observe that $K(Y^*) \leq
(y^*)^2 K(\widehat{Y})$ and that by the same arguments as in our
previous calculation of $\int_{[0,\infty)} \lambda \,d\widehat{Y}$ we
have
\begin{align*}
  K(\widehat{Y}) & = \widetilde{K}(\widehat{Y})=\frac{1}{2}
  \int_{[0,\infty)} \widetilde{\kappa}_t \,
  d\left((\partial\widehat{\Lambda}_{\widetilde{\kappa}_t})^2\right)
=\frac{1}{2}\int_0^{\widetilde{\kappa}_0} (\partial
\widehat{\Lambda}_l)^2 \,dl<\infty.
\end{align*}

We next show that $X^*$ and $Y^*$ are optimal, respectively, for
problem~\eqref{eq:3} and problems~\eqref{eq:7} and~\eqref{eq:8}. In
fact, due to Theorem~\ref{thm:2}, $Y^* = (y^* \partial
\widehat{\Lambda}_{\widetilde{\kappa}}) \vee \eta_0$ satisfies the
first order condition~\eqref{eq:9} and, by Proposition~\ref{pro:2}, is
thus optimal for the convexified problem~\eqref{eq:8} provided that
$Y^*$ is also contained in $\widetilde{\mathcal{Y}}$.  To see that
even $\int_{[0,\infty)} \widetilde{\lambda} \,dY^*=x$ and to deduce
the optimality of $Y^*$ also for problem~\eqref{eq:7} (and thus, by
Proposition~\ref{pro:0}, optimality of $X^*$ for the original
problem~\eqref{eq:3}) it suffices by Theorem~\ref{thm:1} to check
$\{dY^*>0\} \subset \{\lambda=\widetilde{\lambda}\}$ which, in fact,
is immediate from~\eqref{eq:177}.

The formula for the minimal costs when $\eta_0=0$ is an immediate
consequence of our above computations for $\widehat{Y}$. It thus
remains to show that our optimization problems do not have a solution
if $|\partial \widehat{\Lambda}|_{\mathbf{L}^2}=\infty$. To see this
note that in this case there is, for any sufficiently large $0 \leq
S<T<\infty$, a schedule $X^{S,T} \in \mathcal{X}$ which is optimal for
$\delta^{S,T} \set \delta 1_{[S,T]}$ instead of $\delta$ when
$\eta_0=0$. This follows from our earlier results once we note that
the corresponding concave envelope $\widehat{\Lambda}^{S,T}$ always
has a bounded density because $T<\infty$, and thus a solution to this
finite time horizon problem exists provided its market depth does not
vanish identically. This latter condition clearly holds for
$\delta^{S,T}$ when $T$ is chosen sufficiently large, for otherwise
$\delta \equiv 0$ after some time $S$ which would rule out the
presumed explosion of $\partial \widehat{\Lambda}$ at $0$. Note that
we can futhermore choose $S, T \uparrow \infty$ such that
$\widehat{\Lambda}_{\widetilde{\kappa}}$ coincides with
$\widetilde{\Lambda}_{\widetilde{\kappa}}$ at these points. This
ensures that $\widehat{\Lambda}=\widehat{\Lambda}^{S,T}$ on
$[\widetilde{\kappa}_T,\widetilde{\kappa}_S]$ and hence $|\partial
\widehat{\Lambda}^{S,T}|_{\mathbf{L}^2}^2 \to
\int_0^{\widetilde{\kappa}_0} (\partial \widehat{\Lambda}_{k \wedge
  \widetilde{\kappa}_S})^2\,dk= \infty$ as $T \uparrow \infty$.

Now because
$$
C(X^{S,T}) = \int_{[0,\infty)} \frac{\eta_0}{\rho_t} \,dX^{S,T}_t +
C^0(X^{S,T}) \leq \frac{\eta_0x}{\rho_S} + C^0(X^{S,T})
$$
where $C^0(X)$ denotes the cost of any $X \in \mathcal{X}$ when $\eta_0=0$, we 
obtain
$$
\inf_{\mathcal{X}} C \leq 
\frac{\eta_0x}{\rho_S} + \frac{x^2}{2|\partial \widehat{\Lambda}^{S,T}|_{\mathbf{L}^2}^2}
$$
where we used our formula for the optimal costs $C^0(X^{S,T})$.
Because of our special choice of $S,T$, the second term vanishes for
any fixed $S$ as $T \uparrow \infty$. The first term vanishes for $S
\uparrow \infty$ because $\rho$ has to be unbounded for $\partial
\widehat{\Lambda}_k$ to increase to $\infty$ as $k \downarrow
0$. Indeed: $\partial \widehat{\Lambda}_{0+} = \sup_{k>0}
\widetilde{\Lambda}_k/k = \sup_{k>0} \rho_{\tau_k}$.
\end{proofof}

Finally let us show how Theorem~\ref{thm:0} follows from Corollary~\ref{cor:1}:

\begin{proofof}{Theorem~\ref{thm:0}}
  In view of Corollary~\ref{cor:1} it suffices to show that $ \sup_{0
    \leq t \leq s} L^*_t= \partial
  \widehat{\Lambda}_{\widetilde{\kappa}_s}$, $s \geq 0$. Now, from the
  properties of concave envelopes and because of the left-continuity
  of $\partial \widehat{\Lambda}$ we have for any $0 < k \leq
  \widetilde{\kappa}_0$ that
$$
\partial \widehat{\Lambda}_k = \sup_{l \in [k,\widetilde{\kappa}_0]} \inf_{m \in
  [0,l)} \frac{\widetilde{\Lambda}_m - \widetilde{\Lambda}_l}{m-l} \,.
$$
With the changes of variables $k= \widetilde{\kappa}_s$, $l =
\widetilde{\kappa}_t$, $m=\widetilde{\kappa}_u$ the preceding ratio
turns into the one occurring in~\eqref{eq:4}, accomplishing our proof.
\end{proofof}

\bibliographystyle{plainnat} \bibliography{../bib/finance}

\end{document}